\documentclass[draftcls,onecolumn,12pt]{IEEEtran} 
\usepackage[10pt]{moresize}

\usepackage{amsmath}
\usepackage{amsthm}
\usepackage{amsfonts}
\usepackage{amssymb}
\usepackage{graphicx,dsfont}
\usepackage{xcolor}
\usepackage{multirow}
\usepackage{lscape,subfigure}

\usepackage{dblfloatfix }

\newcommand{\bigone}{\mathds{1}}
\newcommand{\Ex}{\mathbb{E}}
\newcommand{\R}{\mathbb{R}}

\newcommand{\Cset}{\mathcal{C}}
\newcommand{\Ract}{\mathcal{A}}
\newcommand{\qmat}{\mathbf{Q}}

\newcommand{\Bset}{\mathcal{B}}

\newcommand{\Tset}{\mathcal{T}}
\newcommand{\Sset}{\mathcal{S}}
\newcommand{\Dset}{\mathcal{D}}

\newtheorem{theorem}{Theorem}[section]

\newcommand{\Ident}{\mathbf{I}}
\newcommand{\Zblock}{\mathbf{0}}
\newcommand{\nnc}{\text{nnc}}
\usepackage{pgfplots}
\newcommand{\xvec}{\mathbf{x}}
\newcommand{\uvec}{\mathbf{u}}
\newcommand{\vvec}{\mathbf{v}}
\newcommand{\hmat}{\mathbf{H}}
\newtheorem{lemma}[theorem]{Lemma}

\graphicspath{Figs}

\title{Performance Evaluation of Advanced Relaying Protocols in Large Wireless Networks}

\author{Andr\'es~Altieri,~\IEEEmembership{Member,~IEEE}, and~Pablo~Piantanida,~\IEEEmembership{Senior~Member,~IEEE}
    \thanks{A. Altieri is with Universidad de Buenos Aires and CSC-CONICET, Buenos Aires, Argentina. Email:  aaltieri@fi.uba.ar, andres.altieri@conicet.gov.ar.}
     \thanks{P. Piantanida is with Laboratoire des Signaux et Syst\`emes (L2S, UMR8506), CentraleSup\'elec-CNRS-Universit\'e Paris-Sud, Gif-sur-Yvette, France. Email: pablo.piantanida@centralesupelec.fr.} 
  \thanks{This work was partially supported by a Level II Bilateral Cooperation Project of the CNRS, MinCyT and CONICET.
  	 The work of P. Piantanida was supported by the PICS MoSiME from CNRS. All numerical simulations were ran in the TUPAC supercomputer at CSC-CONICET (http://tupac.conicet.gov.ar).} 
}

\begin{document}
\maketitle

\begin{abstract}
This paper studies the performance of some state-of-the-art cooperative full-duplex relaying protocols in the context of a large wireless network modelled using stochastic geometry tools. We investigate the outage behaviour for different cooperative schemes, namely, decode-and-forward, noisy-network coding and mixed noisy-network coding, considering fading, path loss and interference from other sources and relays. 
Due to the high complexity of the network topology and the protocols considered, a closed-form analysis is not possible, so our study is performed through extensive but careful numerical simulations, sweeping a large number of relevant parameters. Several scenarios of particular interest are investigated. In this way, insightful conclusions are drawn regarding the network regimes in which relay-assisted cooperation is most beneficial and the potential gains that could be achieved through it. 
\end{abstract}

\section{Introduction}
\label{sec:intro}

For many years, cooperative wireless communications have been an active area of research, showing promising gains in terms of throughput and reliability. One of the most interesting scenarios is that in which cooperation takes place through the use of wireless relays, which aid a transmitter-receiver pair either in a full-duplex or half-duplex fashion \cite{laneman_cooperative_2004, kramer_cooperative_2005}.  In their seminal work, El Gamal and Cover~\cite{cover_capacity_1979} introduced the main cooperative relaying schemes, and since then, these schemes have been improved upon and new ones have been derived (e.g. see~\cite{kramer_cooperative_2005, lim_noisy_2011, Gastpar2011} and references therein).  In many cases, the evaluation of the performance of these protocols via information-theoretic tools are restricted  to the context of uncorrelated additive white Gaussian noises present at the receivers, and only a limited number of users is considered. However, in a more realistic scenario, there will be a large number of source-destination pairs transmitting and the impairments introduced by wireless interference will be much larger than the effects of noise. Furthermore, interference at the receivers will be correlated because the interference comes from the same sources. Finally, in a large wireless network users interact and may cause adverse interference conditions to each other.
In this context, stochastic geometry~\cite{stochastic_geometry2009, haenggi_stochastic_2009} has emerged as a useful tool to model and study different aspects of large wireless networks, allowing the development of closed-form solutions to many interesting problems. Nevertheless, as more complex network architectures, protocols and interactions between the nodes are considered, closed form or approximate solutions become increasingly hard to obtain or even prohibitive. As a matter of fact, as the network topology becomes more complex, the joint distribution of the interferences at several points of the network cannot be characterized. Another common situation is that as more complex communication protocols are considered, the complexity and the number of error events involved in the decoding procedure grow very fast. This is specially true in the context of advanced relaying protocols which may use many relays and sophisticated decoding strategies. For this reason, in some of these complex scenarios it may be interesting to perform numerical simulations, with the aim of gaining insight and drawing qualitative conclusions regarding certain problems.

In this paper, we propose to study the performance of advanced relaying protocols in the context of a large wireless network modeled using stochastic geometry tools. The network is composed of source-destination pairs which attempt to communicate with the help of full-duplex relays. The transmissions in the network are affected by path-loss and slow fading, and each source and its relays cause interference to other relays and destinations in the network. This implies that, as more relays are added, the overall interference increases. Also, the relays that help each source are drawn from a spatial model, meaning that as more relays help a source-destination pair, the further away they will be and their contribution to improving the quality of the links will (on average)   diminish. These two simultaneous effects introduce a balance between cooperation and interference which would also be present in a real network. 
The main metric to evaluate the performance of the protocols is the \emph{outage probability} (OP) that is, the probability that, due to instantaneous conditions, the channel cannot support the rate attempted by a transmitting user.  The protocols considered are decode-and-forward \cite{cover_capacity_1979}, noisy network coding \cite{lim_noisy_2011}, which is an extension of the well-known compress-and-forward scheme \cite{cover_capacity_1979}, 
and mixed noisy-network  coding \cite{Behboodi2015}, which combines both protocols, allowing some of the relays to perform decode-and-forward and others noisy network coding. 
 The transmissions in the network are decoded treating other transmissions, which are not helpful, as noise, and each destination can choose which subset of its helping relays to use to decode the message from its source.
We study the performance of these protocols under different network setups, considering relevant parameters such as relay density, position, and transmission power, and conclusions are drawn about the possible gains that could be achieved, and which scenario is the most favorable.

\subsection{Related Works}

There have been some  works investigating cooperative communication with  relays in large wireless networks via stochastic geometry tools. In \cite{Crismani2015} and  \cite{Tanbourgi2013} the authors considered an outage and diversity-order analysis of a half-duplex selection decode-and-forward protocol in which the interference comes from a network modelled as an homogeneous Poisson Point Process (PPP). In \cite{Altieri_JSAC} the OP of full-duplex and half-duplex decode-and-forward and compress-and-forward schemes for a single relay affected by interference of a Poisson network of interferers was also investigated. In these scenarios it is assumed that only a single relay is active at once, either in full-duplex or half-duplex fashion. Moreover the model lacks of generality since it does not consider that interference is also  caused by relays of other users. Therefore, the analysis performed does not take into account that in real-world networks as more users request  relays there will be an overall increase of interference up to the point where cooperation may become useless or even detrimental. Beside this the relays' positions are assumed to be fixed and in some cases superimposed which does not consider that as more relays are added to communicate, they will likely be further away from the source and destination, reducing the possible improvement of cooperation. 

In \cite{ARVPG_2013} some of the above mentioned issues have been addressed   where the authors  considered a decentralized wireless network in which each source-destination pair have a relay that can be active or not according to some probability. Cooperation can only take place using decode-and-forward, and the authors derive the optimal relay activation probability and the gains that can be achieved in terms of the OP according to the position of the relays. The main focus of these works was  at finding closed form results which are afterwards validated through simulation.  Although whenever it is possible closed-form expressions are very valuable, this approach appears to be highly restrictive since it requires that the network topology and cooperative protocols to be much simpler than the ones studied in the present work, which aims at drawing conclusions in more complex scenarios directly by simulation.  The setup proposed in this paper can be considered as an extension of the preliminary one first investigated in~\cite{ARVPG_2013}. However, in the present work each source-destination pair can take advantage of several different relays and the cooperative protocols considered are the state of the art and thus these are expected to perform better. Furthermore, several additional networks parameters are introduced and studied. 




\subsection{Main Contributions}
The main contribution of this work is studying the performance of some state-of-the-art full-duplex cooperative relaying protocols in a large interference-limited wireless network. This setup is very different to the usual AWGN case in which communications are hampered by uncorrelated Gaussian noise at each receiver. In our setup, additional relays create more interference in the network, and also, since the relays come from a spatial model, activating more relays implies that the relays will on average be further away and the benefits of their activation will be smaller. In this framework, we consider three representative protocols, namely, opportunistic decode-and-forward (ODF)~\cite{1362898}, noisy network  coding (NNC)~\cite{lim_noisy_2011} and mixed noisy-network coding (MNNC)~\cite{Behboodi2015}. For each of these protocols we consider two versions: a standard version in which all the relays of each source-destination pair can transmit independently of their channel qualities towards the source or destination, and \emph{interference aware} versions, in which the relays can transmit only if their channel towards its destination or its source are above a certain threshold. These interference aware protocols aim at reducing the interference in the network by turning off relays which most likely would not help their corresponding source-destination pairs. Since closed forms or close approximations are almost impossible to obtain, we have performed extensive simulations on the network and provide conclusions regarding the dependence of the OP of each protocol with respect to the relay density, the relative transmission powers between relay and sources, the number of active relays for each source-destination pair and the position of the relays. Also, we compare the performance of each of the protocols in their standard and interference aware forms. Testing these complex network models and protocols is a very difficult task that requires large-scale computer-based simulations whose intrinsic difficulty should not be underestimated. This is because the outage events become increasingly complex for more advanced protocols, and also the number of events grows exponentially with the number of relays. Finally, the interference time signals are correlated by the spatial distribution of the nodes in a very complex fashion.

The paper is organized as follows. In Section \ref{sec:model}, we present the mathematical model of the network, and in Section \ref{sec:rates} we present the protocols considered and their outage events. 
In Section \ref{sec:outage_ev}, we show how the outage events are evaluated for the network, and in Section \ref{sec:numerical} we present the numerical results. Finally, in Section \ref{sec:conclusions} we summarize some conclusions and comments.
 
\subsubsection*{Notation} 
$^T$ denotes transpose and $^*$ complex conjugation for scalars and transpose-conjugate for matrices or vectors. $I(\cdot;\cdot)$ and $I(\cdot;\cdot|\cdot)$ denote mutual information and conditional mutual information respectively~\cite{cover_elements_2006}. $h(\cdot, \cdot)$ denotes differential entropy~\cite{cover_elements_2006}.

\section{Network Model} \label{sec:model}

In the sequel we introduce the network model and our main assumptions. We consider a planar network model in which source nodes attempt to communicate a message to their destination with the cooperation of other nearby nodes which act as relays.  Relays are assumed to work in full-duplex mode on the same time slots and frequency bands as the sources. The main modelling assumptions are the following:
\begin{itemize}
	\item The spatial distribution of the sources is modeled as a homogeneous PPP  $\Phi_s$ of intensity $\lambda_s$. Each source has a destination which is located at a distance $D$ from the source in a random uniform direction from the source.
	
	\item Each source-destination pair has a set of $n_r$ potential relays. In order to define the position of the relays we propose the following model: we assume that the nodes which may act as relays are distributed in space as a homogeneous PPP of intensity $\lambda_r$. Then, for each source-destination pair a point, which lies on the line between them, is chosen. For example, for a source at $x$ and a destination at $d_x$, a point $c_x$ is chosen as:
	\begin{equation}
	c_x = x + \varepsilon \left(d_x-x\right), \label{eq:center}
	\end{equation}
	where $0 \leq \varepsilon \leq 1$ (the same for all sources). Finally, the $n_r$ potential relays of the pair are found as the points of $\Phi_r$ which are closest to $c_{x}$. Notice that $\varepsilon$ is a network setup parameter which allows us to control whether the relays will be chosen closer to the source or the destination on average. 
		
This model involves  the intrinsic difficulty that, since all the relays come from $\Phi_r$, it is possible that  different source-destination pairs --which are close to each other-- will choose the same potential relay. However, this event is very unlikely provided that the density of potential relays is much larger than the density of sources ($\lambda_r \gg \lambda_s$), but it complicates the implementation of the network because this  possibility of sharing relays has to be considered. For this reason, we introduce the following simplifying assumption: the marginal distribution of the position of the potential relays of each cluster are the same as if selected from $\Phi_r$ by the above procedure, but are independent among clusters. With this assumption, we can determine the density of the positions of the potential relays for each cluster and then generate the potential relays for each source-destination pair independently, without needing to explicitly draw the PPP $\Phi_r$, and exhaustively search for the potential relays for each source-destination pairs.
We refer to each group formed by a source-destination pair and their corresponding $n_r$ potential relays a \emph{cluster}. The following lemma gives this distribution and allows the simulation of $n_r$ closest neighbours of each source-destination pair:

	\begin{lemma}
		Let $X_1, X_2, \ldots$ be the positions of the nearest nodes of a homogeneous PPP in $\R^2$ to a fixed point $x$, relative to this point, in order of increasing distance. Then
		$\{\|X_1\|^2, \|X_2\|^2, \ldots \}$
		forms a homogenous PPP of intensity: $\lambda_r \pi$ in $(0, \infty)$, and the phases of these points are independent uniform random variables in $[0,2\pi)$, independent of the process of distances.
	\end{lemma}
\begin{proof}
	See for example \cite{Resnick_adventures}.
\end{proof}
	\item The sources and the relays use Gaussian signaling, that is, the codebooks used are generated as draws of independent complex zero-mean Gaussian random variables of variance $P_s$ for sources and $P_r$ for relays. The communication channels are narrow-band with flat-fading, and transmissions are attenuated both by path loss and independent Rayleigh fading, that is, the channel between to points $x$ and $y$ is:
	\begin{equation}
	g_{x,y} = h_{x,y} \sqrt{l(x,y)},
	\end{equation}
	where $h_{x,y}$ is a complex circularly symmetric Gaussian (CCSG) signal fading coefficient with zero-mean and unit variance; $l(x,y) \triangleq \|x-y\|^{-\alpha}$ denotes the power path loss function ($\alpha >2$), and the channels are independent between points. This means that the power fading coefficients $|h_{x,y}|^2$ are independent unit-mean exponential random variables.

	We assume that a destination is located at the origin, with its source located at $p_{\mathbf{s}} = (-D, 0)$, and they are called the typical destination and source, respectively.  The $n_r$ potential relays of the typical source-destination pair are centered around the point $c_\mathbf{s}$ according to (\ref{eq:center}) and are denoted, in order of distance to this point, as $\{r_1, \ldots, r_{n_r}\}$. We call the ``typical cluster'' the group of users formed by the destination at the origin, its source and its potential relays.
	Channel gains from the source to its potential relays and the destination are denoted as $g_{s,i}$ ($i=1,\ldots,n_r$) and $g_{s,d}$, respectively. Channel gains from the potential relays to the other relays and the destination are denoted as $g_{i,j}$ and $g_{i,d}$, ($i,j=1,\ldots,n_r$, $i\neq j$), respectively. Channel gains from an interfering source at $x$ to the typical relays and the destination  are denoted as $\tilde{g}_{x,i}$, $i=1,\ldots, n_r$, and  $\tilde{g}_{x,d}$, respectively. Finally, the channel gains from the $i$-th relay of a source at $x$ to the $j$-th relay and the destination of the typical cluster are denoted as $\tilde{g}_{x,i,j}$, and $\tilde{g}_{x,i,d}$, $i,j=1,\ldots,n_r$, respectively. In Fig. \ref{fig:netrep} we can see a representation of the network and the relevant channels.
\end{itemize}
%
\begin{figure}
	\centering
	\includegraphics[width=.8\linewidth]{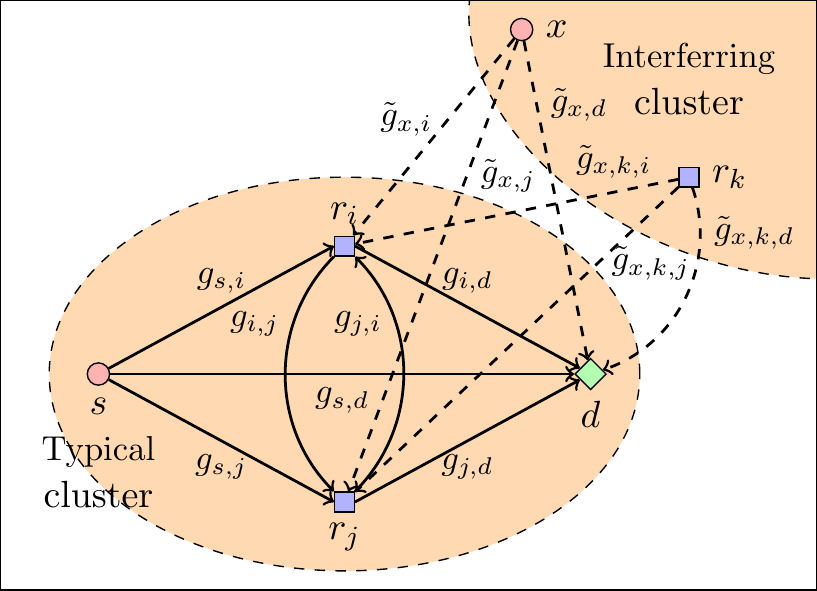}
	\caption{Representation of the network and relevant channel involved. $d$ denotes the typical destination at origin, and $s$ is the typical source at $p_\mathbf{s}=(-D,0)$.}
	\label{fig:netrep} 
\end{figure}
In the interference aware protocols, a destination may have some of its $n_r$ potential relays turned off, if their channels toward their destination or from their source are not considered to be strong enough. For a source located at $x$ we denote by $\Ract_x \subseteq \{1,\ldots, n_r\}$ the set of the indexes of the potential relays that have not been turned off due to bad channels while for the typical cluster, we denote this set by $\Ract_s$. 
Some of these relays which have not been turned off will be transmitting, depending on the protocol employed. For example, in ODF, only the relays which have not been turned off and can decode the message from the source will transmit. We denote by $\Bset_x\subseteq \Ract_x$ the relays of the source at a point $x$ which are transmitting, and by $\Bset_s$ the relays of the typical cluster which are transmitting.

With this model we can now determine the signals received in each node of the typical cluster, and the power and correlation of their corresponding interference signals. We denote by $Y_{d,k}$ and $Y_{i,k}$ the signals received at the destination and the $i$-th relay of the typical cluster at the time instant $k$, respectively, which we may write as:
\begin{align}
Y_{d,k} &= g_{s,d} X_{s,k} + \sum_{m=1}^{n_r} \bigone_{\{m \in \Bset_s\}} g_{m,d} X_{m,k} + Z_{d,k},\\
Y_{i,k} &= g_{s,i} X_{s,k} + \sum_{m=1, m\neq i}^{n_r} \bigone_{\{m \in \Bset_s\}} g_{m,i} X_{m,k} + Z_{i,k}, \label{eq:Yik}
\end{align}
where $X_{s,k}$ and $X_{m,k}$ denote the symbols transmitted by the typical source and its $m$-th relay at time $k$. $Z_{d,k}$ and $Z_{i,k}$ are the interference time signals, which we may write as:
\begin{align}
Z_{d,k} &= \sum_{x \in \Phi_s} \tilde{g}_{x,d} \tilde{X}_{x,k} +  \sum_{m=1}^{n_r}  \bigone_{\{m \in \Bset_{x}\}} \tilde{g}_{x,m,d} \tilde{X}_{x,m,k}, \\
Z_{i,k} &= \sum_{x \in \Phi_s} \tilde{g}_{x,i} \tilde{X}_{x,k} +  \sum_{m=1}^{n_r} \bigone_{\{m \in \Bset_{x}\}} \tilde{g}_{x,m,i} \tilde{X}_{x,m,k},
\end{align}	
where $\{\tilde{X}_{x,k}\}$ and $\{\tilde{X}_{x,m,k}\}$ are the symbols transmitted by the interfering sources and their relays at time $k$. If we condition on the point process and the fading coefficients, under the independent Gaussian signalling hypothesis, the received signals and the interferences are Gaussian, and their distribution is the same for every time instant. So, conditioning, we can find the conditional variance and correlation between the interference time signals $\left\{Z_{d,k}, Z_{1,k}, \ldots , Z_{n_r,k}\right\}$ to fully characterize their joint conditional distribution. Thus, the random interference powers at the nodes of the typical cluster are (for every time instant):
\begin{align}
I_d &\triangleq \Ex\left[|Z_{d,k}|^2\right] \\& =  \sum_{x \in \Phi_s} |\tilde{g}_{x,d}|^2 \Ex\left[|\tilde{X}_{x,k}|^2\right] +  \sum_{m=1}^{n_r}  \bigone_{\{m \in \Bset_{x}\}} |\tilde{g}_{x,m,d}|^2 \Ex \left[|\tilde{X}_{x,m,k}|^2\right] \\
&= \sum_{x \in \Phi_s} \left[P_s |\tilde{g}_{x,d}	|^2  + P_r \sum_{m=1}^{n_r}  \bigone_{\{m \in \Bset_{x}\}} |\tilde{g}_{x,m,d}|^2\right], \\
I_{r_i} &\triangleq \Ex\left[|Z_{i,k}|^2\right] \\&=  \sum_{x \in \Phi_s}\left[P_s |\tilde{g}_{x,i}|^2  + P_r \sum_{m=1}^{n_r} \bigone_{\{m \in \Bset_{x}\}} |\tilde{g}_{x,m,i}|^2\right] 
\end{align}
with $i=1,\ldots,n_r$.
The correlation between the interference time signals is (for any time instant):
\begin{align}
\hspace{-4mm}\beta_{r_i,r_j} &\triangleq \Ex \left[ Z_{i,k} Z_{j,k}^* \right] \\
& = \sum_{x \in \Phi_s} \tilde{g}_{x,i}\tilde{g}_{x,j}^* \Ex\left[|X_{x,k}|^2\right] +  \sum_{m=1}^{n_r} \bigone_{\{m \in \Bset_{x}\}} \tilde{g}_{x,m,i} \tilde{g}_{x,m,j}^*  \Ex\left[|X_{x,m,k}|^2\right]\\
& = \sum_{x \in \Phi_s} \hspace{-1mm} \left[P_s \tilde{g}_{x,i}\tilde{g}_{x,j}^*  + P_r \sum_{m=1}^{n_r} \bigone_{\{m \in \Bset_{x}\}} \tilde{g}_{x,m,i} \tilde{g}_{x,m,j}^*\right] \ \ \ i,j=1,\ldots,n_r, i \neq j,
\end{align}
\begin{align}
\beta_{r_i,d} &\triangleq \sum_{x \in \Phi_s} \tilde{g}_{x,i} \tilde{g}_{x,d}^* \Ex\left[|X_{x,k}|^2\right]  +  \sum_{m=1}^{n_r} \bigone_{\{m \in \Bset_{x}\}} \tilde{g}_{x,m,i}\tilde{g}_{x,m,d}^* \Ex\left[|X_{x,m,k}|^2\right]\\
&= \sum_{x \in \Phi_s}  \left[P_s \tilde{g}_{x,i} \tilde{g}_{x,d}^*  + P_r \sum_{m=1}^{n_r} \bigone_{\{m \in \Bset_{x}\}} \tilde{g}_{x,m,i}\tilde{g}_{x,m,d}^*\right] \ \ \ i=1,\ldots,n_r. 
\end{align}

\section{Information-Theoretic Rates of Selective Cooperative Relaying} \label{sec:rates}

In this section we describe the protocols under study and their corresponding achievable rates from an information-theoretic perspective. 
We define the following sets:
\begin{itemize}
\item $\Dset \subseteq \Ract_s$: relays from the typical cluster that can decode the transmission from the typical source  treating the transmissions of all the other typical relays and from the other clusters as noise.
\item $\Ract_s \setminus \Dset$ is the  set of typical relays that cannot decode the transmission from the source.
\end{itemize}
We assume that the source is unaware of the channel coefficients towards the destination and the relays, meaning that it cannot adapt its transmission rate, and has to choose a rate $R$ and attempt to communicate at that rate. For a chosen protocol, the transmission will fail and an outage will be declared whenever $R$ is larger than the rate that is achievable for the particular realization of the network. We now define the protocols and their respective outage events:
\begin{itemize}
	\item \emph{Opportunistic Decode-and-Forward~\cite{1362898}}: in this protocol, the relays which can decode the transmission from the source, cooperate as a set of distributed antennas. The potential relays in $\Ract_s \setminus \Dset$ remain inactive, that is, $\Bset_s=\Dset$. In this way an outage will be declared whenever the attempted rate $R$ does not satisfy:
	\begin{equation}
	R < I(X_s, X_\Dset;Y_d), \label{eq:Odf_outage}
	\end{equation}
	where $X_\Dset$ is a set of random variables, one for each relay in $\Dset$, such that $\{X_s, X_\Dset\}$ has the same  joint distribution (Gaussian) as the transmitted symbols of the source and the relays in $\Dset$. 
	
	For determining which relays of the typical cluster belong to $\Dset$ we consider that all the relays and source in the network, including those of the typical cluster, are transmitting, and we determine which relays of the typical cluster can decode the transmission from the typical source under this condition. That is, 
	the $i$-th relay $\Ract_s$ will belong to $\Dset$ whenever the rate satisfies:
	\begin{equation}
	R < I(X_s;Y_i), \label{eq:rel_dec_cond}
	\end{equation}
	where $Y_i$ is the received signal at the relay when all the relays in the network (in $\Ract_s$ and $\Ract_x$, $\forall x\in \Phi_s$) and sources are transmitting, and their transmissions are treated as noise for decoding.  Also, to evaluate the outage event (\ref{eq:Odf_outage}) we consider that all the relays outside the typical cluster are transmitting ($\Ract_x=\Bset_x$, $\forall x\in \Phi_s$). This effectively gives an upper bound to the best OP attainable through ODF; this is because each time a relay is set to transmit, the interference in the network changes, and some relays which were previously decoding may not be able to decode any more. This means that a network coordination strategy could be considered in which a different (optimal) number of relays could be active in each cluster to minimize the OP. However, this is impossible to implement or consider because of the large number of relays in the network and the large-scale coordination that would be required.
	
	
	\item \emph{Noisy Network Coding~\cite{lim_noisy_2011,7265060}:} in this protocol, all the relays in $\Ract_s$ perform NNC, even if they can decode the message from the source, that is $\Bset_s = \Ract_s$. Then the destination can try to decode the message by using the transmission of the source and any combination of the relays in $\mathcal{A}_s$ (treating the rest as noise). Thus  outage event is declared when the attempted rate $R$ does not satisfy:
	\begin{equation}
	R < \max_{\Tset  \in  2^{\Ract_s}}\min_{\Sset \in 2^\Tset} \{ I(X_s, X_{\Sset}; \hat{Y}_{\Sset^c},Y_d|X_{\Sset^c}) - I(\hat{Y}_\Sset; Y_\Sset | X_s, X_\Tset,\hat{Y}_{\Sset^c},Y_d)\}, \label{eq:rateNNC}
	\end{equation}
	where $2^{\Ract_s}$ denotes the power set of $\Ract_s$.
	$\hat{Y}_{\Sset}$ and 	$\hat{Y}_{\Sset^c}$ are sets of random variables which are constructed from the received signals of the typical relays $Y_{1},...,Y_{n_r}$ as follows: for each $Y_i$ we define:
	\begin{equation}
	\hat{Y}_i = Y_i + Z_{c,i}, \label{eq:noise_add}
	\end{equation}
	where $Z_{c,i}$ is a complex circularly symmetric white Gaussian noise of zero-mean and variance $n_c$ independent of everything else. The addition of noise $Z_c$ acts as a compression on the received signal $Y_i$. Since conditioned on the point process and the fading coefficients the random variables $\{Y_{1},...,Y_{n_r}\}$ are jointly Gaussian, then $\{\hat{Y}_1,...,\hat{Y}_{n_r}\}$ are also Gaussian.	
	 From (\ref{eq:rateNNC}) we see that each set $\Sset$ is a subset of $\{1,...,n_r\}$; then, for each of the indexes in $\Sset$ the random sets $\hat{Y}_{\Sset}$ and $\hat{Y}_{\Sset^c}$ are made of the $\hat{Y}_i$ whose indexes belong to $\Sset$ and $\Sset^c$, respectively. The sets $X_{\Sset}$ and $X_\Tset$ are constructed as $X_{\Dset}$ in the ODF protocol.
	\item \emph{Mixed Noisy Network Coding~\cite{Behboodi2015}:} this protocol is a combination of the two above. The relays in $\Dset$, which can decode the transmission from the source, cooperate as a set of distributed antennas, while the other ones in $\Ract_s\setminus \Dset$ use NNC. In this case, determining $\Dset$ by the procedure indicated in ODF does not provide an upper bound to the OP, because all the relays in the network are always transmitting. The destination then chooses which subset of $\Ract_s\setminus \Dset$ to use for decoding the message treating the rest as noise. Thus the attempted rate $R$ has to satisfy:
\begin{equation}
R < \max_{\Tset  \in  2^{\Ract_s \setminus \Dset}}\min_{\Sset \in 2^\Tset} \{ I(X_s, X_\Dset, X_{\Sset}; \hat{Y}_{\Sset^c},Y_d|X_{\Sset^c}) - I(\hat{Y}_\Sset; Y_\Sset | X_s,X_\Dset, X_\Tset,\hat{Y}_{\Sset^c},Y_d)\}.
\end{equation}	
All the definitions are the same as in the previous protocols.
\end{itemize}
For each of these basic protocols we consider two versions according to how the relays are activated:
\begin{itemize}
\item \emph{Standard versions of ODF, NNC and MNNC:} all the potential relays in a cluster may transmit, meaning that $\Ract_x = \{1,\ldots,n_r\} = \Ract_s$ for all $x\in\Phi_s$.
\item \emph{Interference aware versions: }a threshold activation  scheme for activating the relays in all clusters is employed, and this threshold can be used in the source-relay or relay-destination channels. These thresholds attempt to mitigate the interference by turning off relays which may not improve the performance of their cluster. If the source-relay threshold is active, then each relay will be active if the channel from its source exceeds a predefined threshold. On the other hand, if the relay-destination threshold is active, then each relay will be active provided that  its channel towards its destination exceeds the threshold. 

\end{itemize}

\section{Procedure for the numerical evaluation of outage events} \label{sec:outage_ev}

In what follows we describe the outage events corresponding to each of the protocols evaluated for the network model detailed in the previous sections. The procedure to perform the Monte Carlo simulation of the outage probabilities of each protocol  consists in drawing multiple realizations of the network parameters, i.e., node positions and fading coefficients, then finding the mutual informations which appear in the outage events and  checking if the selected rate $R$ is above or below the achievable rate given for each realization.
In order to do this, we perform the following procedure:
\begin{enumerate}
	\item We draw a realization of the network and the typical cluster according to the model in Section \ref{sec:model}.
	\item Given the network realization and the chosen protocol, the mutual informations in the outage events for Section \ref{sec:rates} can be computed and the outage condition can be determined. To do this, the mutual informations are written in terms of differential entropies using the following standard identities~\cite{cover_elements_2006}:
		\begin{align}
		I(X;Y) &= h(X) + h(Y) - h(X,Y), \label{eq:minde1}\\
	I(X;Y|Z) &= h(X,Z) + h(Y,Z) - h(X,Y,Z) - h(Z). \label{eq:minde2}
		\end{align}
		where $X, Y, Z$ are continuous random variables and $h(\cdot)$ denotes differential entropy~\cite{cover_elements_2006}.	 
	Since the nodes use Gaussian signaling, conditioned on the realization of the network, all the random variables involved in the mutual informations are CCSG random variables. It is well known that for a CCSG random vector $\xvec$ with covariance matrix $\qmat_{\xvec}$, the entropy is~\cite{Telatar99}:
	\begin{equation}
	h(\xvec) = \log \det (\pi e \qmat_{\mathbf{x}}). \label{eq:covmat1}
	\end{equation}
Therefore, to determine if an outage event occurs we must compute the mutual informations in the outage events, which are written in terms of joint entropies and require all the covariance matrices of the random variables appearing in the joint entropies. 
	\end{enumerate}
The large computational burden of this procedure can be reduced by noticing that it is not necessary to calculate a covariance matrix for each differential entropy. If we compute the covariance matrix of all the random variables involved in the typical cluster and the chosen protocol  we can find the other ones which are required by deleting some of its rows and columns. This is the covariance matrix with largest dimentions for a given setup, and it is necessary to find it to calculate some of the mutual informations involved.
Also, another way to reduce the computational burden of the procedure for NNC and MNNC is to notice that if we find a set $\Tset  \in  2^{\Ract_s}$ such that the attempted rate exceeds the minimum corresponding to that set in the outage condition, then it is sufficient to declare an outage and no further computation is required.

In what follows we focus on the typical cluster and describe how to find the largest covariance matrix of all the random variables involved for a chosen protocol.
First, from the $n_r$ potential relays of the source, we have a subset $\Ract_s$ of the relays which can transmit (all of them if the protocol is not interference aware).
Among these, only a subset $\Dset \subseteq \Ract_s$ will be able to decode the message of the source and act as secondary antennas. To determine those relays, we must evaluate for which relays condition (\ref{eq:rel_dec_cond}) is met. This condition using (\ref{eq:minde1}) and (\ref{eq:covmat1}), can be written for the $i$-th relay as:
$$R < \log_2 \left( 1 + \frac{|g_{s,i}|^2 P_S}{I_{r_i} + \sum_{m=1, m\neq i}^{n_r} \bigone_{\{m \in \Ract_s\}} |g_{m,i}|^2 P_r}\right).$$
The relays which do not fulfill this condition (cannot decode) may perform NNC (if the NNC or MNNC protocols are employed) or remain silent. Let us define the set of relays which perform NNC as:
\begin{equation}
\Cset_s \triangleq \{t_1, \ldots, t_{N_{\nnc}}\} \subseteq \Ract_s \setminus \Dset.
\end{equation}
Notice that we cannot say that $\Cset_s = \Ract_s \setminus \Dset$ because for example in ODF, $\Cset_s= \emptyset$ while $\Ract_s \setminus \Dset$ will have the relays that cannot decode the message from the source. We define as:
\begin{equation}
\mathcal{D} \cup \Cset_s = \{u_1, \ldots, u_{N_a}\} \subseteq \{1, \ldots, n_r\}.
\end{equation}
the set of relays of the typical cluster which will be transmitting using either ODF or NNC. In order to find the correlation matrix of the random variables appearing in the outage events, we define a random vector that has the signals which are transmitted and received by the nodes of the typical cluster, and also the noise signals which are added at each relay according to (\ref{eq:noise_add}) for NNC: 
\begin{equation}
\uvec \triangleq \left[X_{s},X_{u_1} , \ldots , X_{u_{N_a}},Z_{u_1},\ldots, Z_{u	_{N_a}} , Z_{d} , Z_{c, t_1} ,\ldots ,Z_{c,t_{N_{\nnc}}} \right]^T.
\end{equation}
We do not include the time instant, since, as mentioned before, the distribution of the vector does not depend on it. In the case of ODF, since $\Cset_s=\emptyset$ there is no need for compression noise random variables.
Since all the transmitters in the network use independent Gaussian signaling the correlation matrix $\qmat_{\uvec}$  of the vector $\uvec$ is block diagonal:
\begin{equation}
	\qmat_{\uvec} \triangleq \Ex\left[\uvec \uvec^*\right] = \left[\begin{array}{c|c|c|c}P_s & \multicolumn{3}{c}{\Zblock_{1,2N_a+N_{\nnc}+1} }\\\hline
	\multirow{3}{*}{$\Zblock_{2N_a+N_{\nnc}+1,1}$} & P_r \Ident_{N_a} & \multicolumn{2}{c}{\Zblock_{N_a,N_a+N_{\nnc}+1}} \\\cline{2-4}	
	& \multirow{2}{*}{$\Zblock_{N_a+N_{\nnc}+1,N_a}$}&  \qmat_Z &\Zblock_{N_a+1,N_{\nnc}}\\ \cline{3-4	}
	& &\Zblock_{N_{\nnc},N_a+1}& n_c \Ident_{N_{\nnc}}\end{array} \right], \label{eq:CU}
	\end{equation}
 where:
\begin{itemize}
	\item $\Ident_{n}$ denotes and identity matrix of $n\times n$;
	\item $\Zblock_{n,m}$ denotes a block of zeros of $n\times m$;
\item $\qmat_Z$ is a square matrix of side $N_a+1$ containing the correlation between the interference random variables:\addtocounter{equation}{1}
\begin{equation}
\hspace{-6mm}\qmat_Z \hspace{-.5mm} \triangleq \hspace{-.5mm} \left[ \hspace{-.5mm} \begin{array}{ccccc}
I_{u_1} & \beta_{u_1,u_2} &  \ldots &  \beta_{u_{1},u_{N_a}} & \beta_{u_1,d}\\
\beta_{u_2,u_1} & I_{u_2} &  \ldots &  \beta_{u_{2},u_{N_a}} & \beta_{u_2,d}\\
\vdots & &\ddots & & \vdots\\
\beta_{u_{N_a},u_{1}} & \beta_{u_{N_a},u_{2}}& & I_{u_{N_a}}&   \beta_{u_{N_a},d}\\
\beta_{d,u_1} & \beta_{d,u_{2}} & \ldots  &\beta_{d,u_{N_a}}  & I_d  
\end{array} \hspace{-.5mm} \right]\hspace{-1mm} .\hspace{-4mm}
\end{equation}
\end{itemize}
 Then we define a random vector containing all involved random variables that appear in the outage events:
\begin{equation}
\vvec \triangleq \left[X_{s}, X_{u_1},\ldots,X_{u_{N_a}},Y_{u_1}, \ldots, Y_{u_{N_a}}, Y_{d}, \hat{Y}_{t_1}, \ldots, \hat{Y}_{t_{N_{\nnc}}}\right] ^T. \label{eq:vvec1}
\end{equation}
Now we find the matrix which allows us to write $\vvec$ in terms of $\uvec$.  	
We define a fading coefficient matrix between the source and the relays towards the destination and the relays:
\begin{equation}
\hmat \hspace{-.5mm} \triangleq \hspace{-.5mm}  \left[ \begin{array}{c @{\hspace{0.75\tabcolsep}} c  @{\hspace{0.75\tabcolsep}} c @{\hspace{0.75\tabcolsep}}c @{\hspace{0.75\tabcolsep}}c @{\hspace{0.75\tabcolsep}}c}
h_{s,u_1} & 0 & h_{u_2,u_1} & h_{u_3,u_1} & \ldots & h_{u_{N_a},u_1} \\
h_{s,u_2} & h_{u_1,u_2} & 0 & h_{u_3,u_2} & \ldots & h_{u_{N_a},u_2} \\
\vdots & & & \ddots& & \\
\vdots & & & & \ddots& \\
h_{s,u_{N_a}} & h_{u_1,u_{N_a}} & \ldots  & & & 0
\end{array} \right].
\end{equation}
With this we may write:	
\begin{equation}
\vvec = \tilde{\hmat} \uvec
\end{equation}
with $\tilde{\hmat}$ a matrix:
\begin{equation}
\hspace{-.5mm}\tilde{\hmat} \hspace{-1mm} \triangleq \hspace{-1mm} \left[ \begin{array}{c @{\hspace{0.1\tabcolsep}}| c @{\hspace{0.5\tabcolsep}}| @{\hspace{0.5\tabcolsep}}c| c}
{\Ident_{N_a+1}} & \multicolumn{3}{c}{\Zblock_{N_a+1,N_a+N_{\nnc}+1}} \\\hline
\multirow{2}{*}{$\hmat$} & \Ident_{N_a}&\multicolumn{2}{c}{\Zblock_{N_a,N_{\nnc}+1}} \\ \cline{2-4}
&\Zblock_{1,N_a} & 1 & \Zblock_{1,N_{\nnc}} \\ \hline
\hmat_{ (t_1, \ldots, t_{N_{\nnc}})}  & \Ident_{N_{a}(t_1, \ldots, t_{N_{\nnc}})} &\Zblock_{N_{\nnc},1} & \Ident_{N_{\nnc}}
\end{array}\right]
\end{equation}
where:
\begin{itemize}
	\item $\hmat_{(t_1, \ldots, t_{N_{\nnc}})}$ is a matrix containing the rows of $\hmat$ indicated in the vector $[t_1, \ldots, t_{N_{\nnc}}]$, that is:
	\begin{equation}
	[\hmat_{(t_1, \ldots, t_{N_{\nnc}})}]_{i,j} = \hmat_{t_i,j}. \label{eq:htilde1}
	\end{equation}
	\item $\Ident_{N_{a}(t_1, \ldots, t_{N_{\nnc}})}$ is matrix of size $N_{\text{nnc}}\times N_a$  obtained by taking the identity matrix of size $N_a$ and keeping the rows indicated in the vector $[t_1, \ldots, t_{N_{\nnc}}]$, that is:
\begin{equation}
[\Ident_{N_a(t_1, \ldots, t_{N_{\nnc}})}]_{i,j} = \begin{cases} 1 & \text{if $j=t_{i}$,} \\
0 & \text{otherwise.}\end{cases}
\end{equation}

\end{itemize}
Then the covariance matrix of the vector $\mathbf{v}$ can be found as:
\begin{equation}
\qmat_\vvec = \Ex \left[ (\tilde{\hmat} \uvec)(\tilde{\hmat} \uvec)^*\right] = \tilde{\hmat} \qmat_{\uvec} \tilde{\hmat}^*.
\end{equation}
As mentioned above, the matrix $\qmat_\vvec$ is calculated only once for each realization of the network and is used to find all the entropies required to evaluate the outage events. To evaluate the joint entropy between any of the variables in $\vvec$ one must take the matrix $\qmat_\vvec$ and delete the rows and columns corresponding to the elements of $\vvec$ whose entropy one does not need to calculate. Then, applying (\ref{eq:covmat1}), the joint entropy, and hence the mutual information between any of the variables in $\vvec$ can be found. 

Following this procedure the outage events for the selected protocol can be evaluated and it can be determined if an outage has taken place for the realization of the network.

\section{Numerical Results and Discussions} \label{sec:numerical}

In performing the simulation of the outage probabilities for each of the protocols a large number of parameters were swept ($(\lambda_r, n_r, P_r, n_c)$ among others) over a wide range, and a large number of realizations of the network ($>10^5$) where drawn for each setup. Given this and the complexity of the outage events involved, a standard desktop computer could not be employed to perform the simulations. For this reason, the Tupac supercomputer cluster hosted at CSC-CONICET (http://tupac.conicet.gov.ar) was employed. With these powerful machines, we are able to perform the large-scale simulations required for this complex networks and protocols. 

In what follows we describe the simulation setup. The density of sources is $\lambda_s = 10^{-4}$ nodes/unit area. The destination is located at the origin and its source is located at $p_\mathbf{s} = (-10, 0)$. The relays are chosen as the nearest neighbors of a point $c_{p_\mathbf{s}}$ which lies on the line between the source and the destination, according to (\ref{eq:center}), where $\varepsilon=0$  implies the relays are centered around the source, and $\varepsilon=1$ means they are centered around the destination. The density of potential relays $\lambda_r$ is chosen as a multiple of the density of sources. Each source can use at most $n_{r}$ relays (the same for all clusters). 
The sources transmit with unit power $P_s=1$ and the relays with a fixed power $0 \leq P_r \leq P_s$. For the case of protocols involving NNC, the compression noise variance $n_c$ is optimized for each network setup (the optimized value of $n_c$ was always between$\left[10^{-8}, 10^{-2}\right]$). The attempted rate in all cases is $R=1$ bit/use, and the path loss exponent is $\alpha=4$.
Finally, each Monte Carlo simulation was obtained by averaging at least $10^5$ realizations of the network. We compare the OP of the protocols with a point-to-point transmission without involving  relays in the network, which is~\cite{baccelli_aloha_2006}
\begin{equation}
P_{\text{out,DT}}(R) = 1 - e^{-\lambda_s C T^{2/\alpha} ||p_\mathbf{s}||^2},
\end{equation}
where:
\begin{gather}
T= 2^R-1,\\
C = \frac{2\pi }{\alpha} \Gamma\left(\frac{2}{\alpha}\right)  \Gamma\left(1-\frac{2}{\alpha}\right), \label{eq:C}
\end{gather}
and $\Gamma(z) = \int_0^\infty t^{z-1} 	e^{-t}dt$ is the standard Gamma function.

In what follows we consider three main questions regarding the three protocols: the behavior of the OP as a function of the relative density between sources and relays, the behavior of the OP as a function of the relay transmission power, and the dependence of the OP with the point around which the relays are chosen. After observing this behavior we analyze the OP that can be obtained by optimizing the relay transmission power and using interference aware relays, which turn themselves off if the channel amplitude towards their destination or source do not exceed a predefined threshold. 



\subsection{Dependence of the OP with the Relay Density}
\begin{figure} [!th]
	\centering
\includegraphics[]{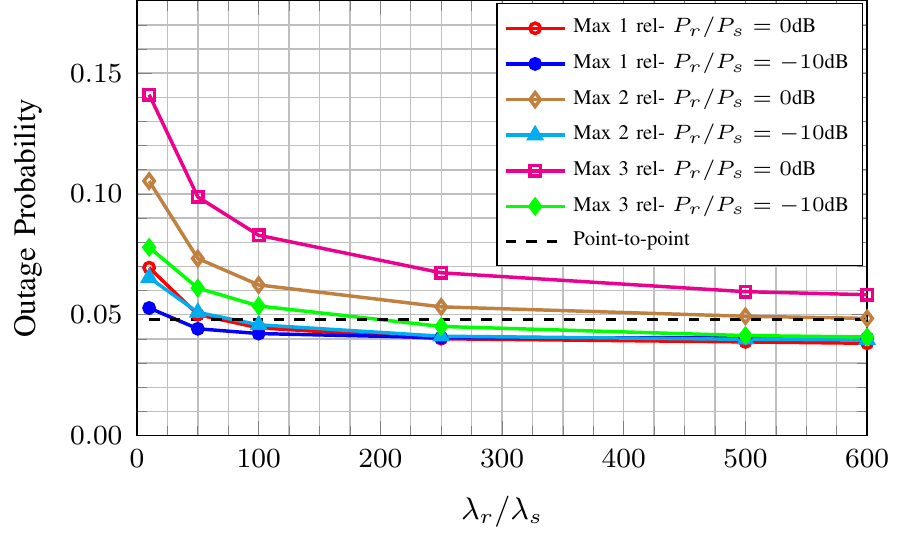}
		\caption{Outage probability for ODF as a function of the relative density between relays and sources, for different relay-source relative transmission powers. $\lambda_s =10^{-4}$, $R=1$b/use, $\alpha=4$.  Relays are chosen centered around the source ($\epsilon = 0$).}
		\label{fig:ODF1}
\end{figure}
\begin{figure}
		\centering
\includegraphics[]{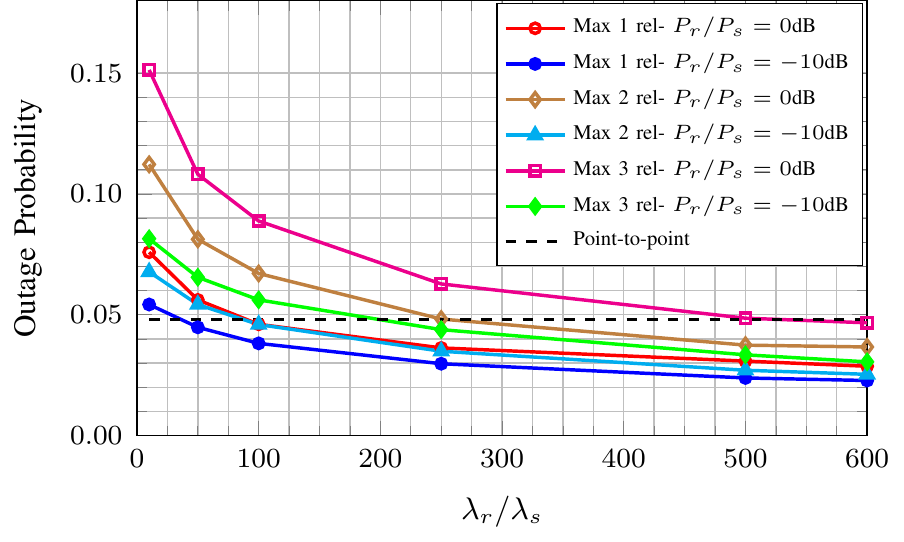}
		\caption{Outage probability for NNC as a function of the relative density between relays and sources, for different relay-source relative transmission powers. $\lambda_s =10^{-4}$, $R=1$b/use, $\alpha=4$. Relays are chosen centered around the destination ($\epsilon = 1$). OP is optimized w.r.t. the noise compression variance.}
		\label{fig:NNC1}
\end{figure}
\begin{figure}\begin{center}
\includegraphics[]{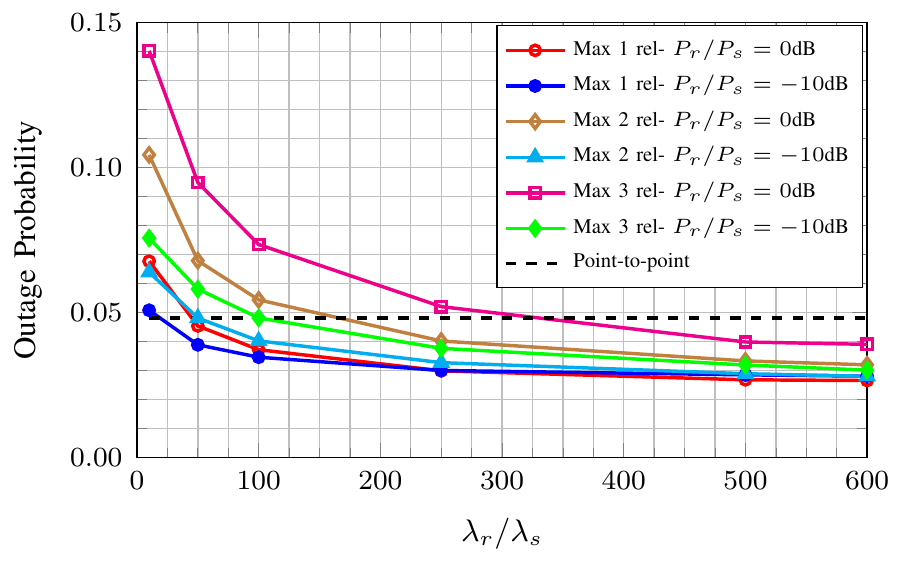}
		\caption{Outage probability for MNNC as a function of the relative transmission power between relays and sources, for different relative densities. $\lambda_s =10^{-4}$, $R=1$b/use, $\alpha=4$. Relays are chosen centered around the midpoint between a source and its destination ($\epsilon = 0.5$). The OP is optimized w.r.t. the noise compression variance.}
		\label{fig:MNNC1}
\end{center}
\end{figure} 

In Figs. \ref{fig:ODF1}, \ref{fig:NNC1} and \ref{fig:MNNC1} we plot the OP as a function of the relative relay density $\lambda_r/\lambda_s$, for different number of relays and fixed relative transmission powers $P_r/P_s$ for ODF, NNC and MNNC, respectively. The relays are chosen centered around the source ($\epsilon=0$) for ODF, around the destination ($\epsilon=1$) for NNC, and in the middle ($\epsilon=0.5$) for MNNC. The values of $\epsilon$ were chosen according to what is supposed to be the best option for each protocol. ODF is expected to work better when the relays are, on average, closer to the source because this increases the chances for relays to decode the transmission, while NNC will perform better if the relays are closer to the destination which receives a compressed version of the observation of the relays. Finally, MNNC is a combination of both protocols and thus, it is expected to  outperform the other  when the relays are on average midway between the source and the destination.

For all the protocols it is interesting to observe that cooperation is more beneficial when the density of potential relays is much larger than the density of sources (100 times or more according to the scenario), that is, the OP is decreasing with $\lambda_r/\lambda_s$ in all cases.  Also, the performance is either improved or does not decrease if the relays use a smaller transmission power than the source ($P_r/P_s=-10$dB in the plot). It is also worth to mention that using more relays does not improve the OP with respect to using a single one. This is because the additional relays in the other clusters increase the interference at the typical clusters, and  because the second and third relays are further away than the first one, so that the benefits of cooperation are reduced by path loss. Furthermore, the destination chooses the best set of relays for decoding and treats the rest as noise. Although cycling through all the combinations of relays improves the chances of decoding, the interference generated by treating the rest of the relays and noise, added to the interference from other clusters, does not result in any benefits in the OP. In the case of ODF, this conclusion may be affected by the assumption that we made in order to keep  the problem tractable that relays in the other clusters always remain on.

Although the OP is decreasing in $\lambda_r/\lambda_s$, the gains, however, are not the same in this regime.
For the case of ODF the gains are not substantial in this regime, that is, when the relays are chosen to be located around the source. Whereas based on NNC and MNNC a reduction of the OP close to 50\% sounds feasible  when the relays are chosen closer to the destination.

Perhaps, the most interesting conclusion from this setup is that cooperation appears to be most useful in networks in which relays from a dense network of low-power nodes (compared to the sources), such as sensor or cellphone networks.

\subsection{Dependence of the OP with the Relay Transmission Power}

In Figs. \ref{fig:ODF2}, \ref{fig:NNC2}, and \ref{fig:MNNC2}  we plot the OP for ODF, NNC and MNNC as a function of the relative transmission power between relays and sources, for different relative relay-source densities. As in the previous section, the relays are chosen centered around the source ($\epsilon=0$) for ODF, around the destination ($\epsilon=1$) for NNC, and in the middle ($\epsilon=0.5$) for MNNC. 
For the case of NNC (Fig. \ref{fig:NNC2}) we see that the OP is increasing in the relay transmission power for moderate or large relay densities. On the other hand, for ODF or MNNC the OP is also increasing in general, except when a single relay is used and the density of relays is large. In that case the OP is decreasing in the relay transmission power, but  the gains are marginal to warrant the increase in relay transmission power. 

As we have seen before, without any optimization in the relay transmission power or position we see that large gains in terms of OP can be achieved through NNC or MNNC while in the case of ODF this setup is not the most convenient.



\begin{figure} [!t]
	\centering
\begin{minipage}[t]{0.49\textwidth} 
\includegraphics[]{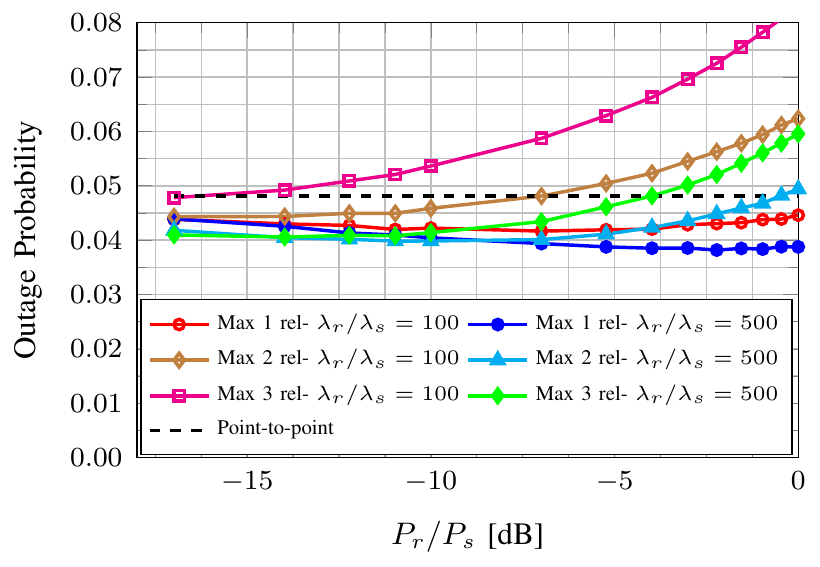}
	\caption{Outage probability for ODF as a function of the relative transmission power between relays and sources, for different relative densities. $P_r/P_s =-10dB$. $\lambda_s =10^{-4}$, $R=1$b/use, $\alpha=4$.  Relays are chosen centered around the sources ($\epsilon = 0$).}
	\label{fig:ODF2}
	\end{minipage}
\begin{minipage}[t]{0.49\textwidth}
	\centering
\includegraphics[]{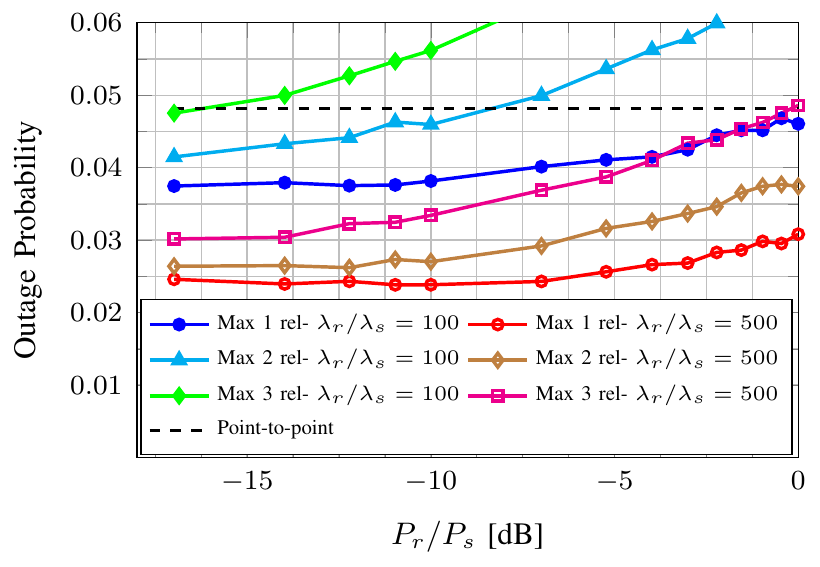}
			\caption{Outage probability for NNC as a function of the relative transmission power between relays and sources, for different relative densities. $P_r/P_s =-10dB$. $\lambda_s =10^{-4}$, $R=1$b/use, $\alpha=4$. $d=10$. Relays are chosen centered around the destination ($\epsilon = 1$). OP is optimized w.r.t. the noise compression variance.}
			\label{fig:NNC2}
\end{minipage}	
\begin{minipage}[t]{0.49\textwidth}
	\centering
\includegraphics[]{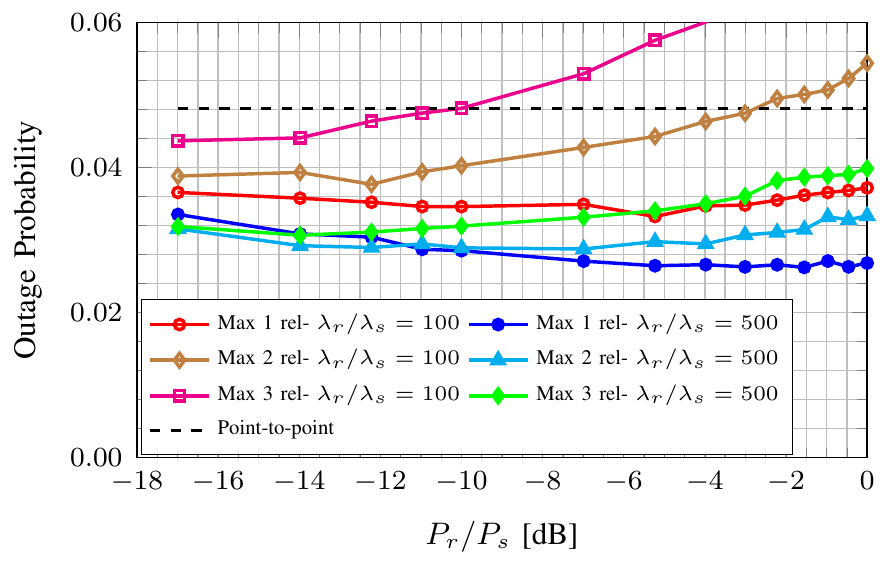}
	\caption{Outage probability for MNNC as a function of the relative transmission power between relays and sources, for different relative densities. $P_r/P_s =-10dB$. $\lambda_s =10^{-4}$, $R=1$b/use, $\alpha=4$. Relays are chosen centered around the midpoint between a source and its destination ($\epsilon = 0.5$). The OP is optimized w.r.t. the noise compression variance.}
	\label{fig:MNNC2}
\end{minipage}
\end{figure} 

\subsection{Dependence of the OP with the Relay Position}
In Figs. \ref{fig:ODF3},  \ref{fig:NNC3}, and \ref{fig:MNNC3}  we plot the OP as a function $\epsilon$, which indicates the center point around which the relays are chosen on the line between the source and the destination, for different relative relay-source densities. In all cases we set the relative transmission power $P_r/P_s = -10$dB, which has shown to be reasonable in previous plots. 

For the case of ODF (Fig. \ref{fig:ODF3}) we see that the biggest gains can be obtained when the center point is chosen near the midpoint between the source and destination, but closer to the destination. This is because at this position, on average, the quality of the source-relay and relay-destination channels are balanced, together with the different source and relay transmission powers, and the OP is minimized. In that case, if a single low-power relay is used, the OP is reduced by more than 40\% when $\lambda_r/\lambda_s \geq 500$, compared to a transmission without cooperation. In addition, a similar reduction can be achieved even if the relative density is smaller ($\lambda_r/\lambda_s \geq 250$). 

For  NNC (Fig. \ref{fig:NNC3}) the OP appears to be decreasing provided that $0 \leq \epsilon \leq 1$, implying  that the best would be to chose the relays centered around the destination. In this case the performance is very similar if one or two low-power relays are used and the dispersion in performance is smaller when compared to ODF. 
The potential gains of NNC are larger than that of ODF, but similar. For example, when $\lambda_r/\lambda_s = 500$ a reduction of more than 50\% in the OP can be achieved by using a single relay. 

In the case of MNNC (Fig. \ref{fig:MNNC3}), in addition to the plots with one relay and  $P_r/P_s=-10$dB, we also plot a curve with $P_s=P_r$, because in Fig. \ref{fig:MNNC2} we saw that this may be better for this protocol when at most one relay is used. The behaviour of the OP is similar to that of NNC in the sense that it is more convenient to chose the relays closer to the destination. 
When using a single relay, the potential gains of MNNC are better but similar to the gains obtained with NNC; for $\lambda_r/\lambda_s \geq 500$ a reduction greater than 55\% in the OP can be achieved when the relay uses a low power. For the case of single relay, when $P_r=P_s$ a larger reduction of the OP is obtained near the source, with a loss of performance near the destination as compared to using a low-power relay. This may be because MNNC is a combination of ODF and NNC; near the destination NNC will be dominant, which does not benefit from setting $P_r=P_s$, while near the midpoint ODF will be dominant and will benefit from $P_r=P_s$.
This results in that the OP remains almost constant when $0.3 \leq \epsilon \leq 1$ with a reduction of 45\% in the OP compared to a point-to-point transmission.
\begin{figure} [!t]
	\centering
\includegraphics[]{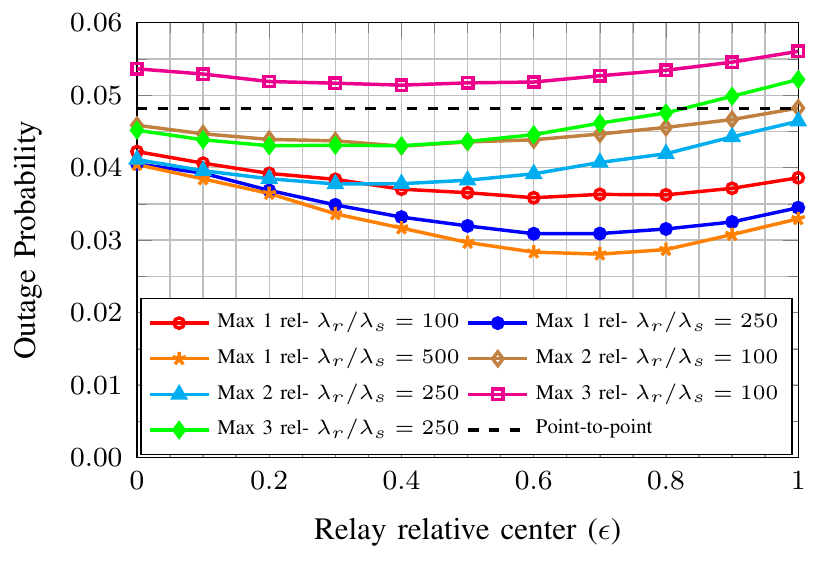}
	\caption{Outage probability for ODF as a function point around which the relays are chosen, one the line between source and destination, for different relative densities. $P_r/P_s=-10$dB, $\lambda_s =10^{-4}$, $R=1$b/use, $\alpha=4$.}
	\label{fig:ODF3}
\end{figure}
\begin{figure}
	\centering
\includegraphics[]{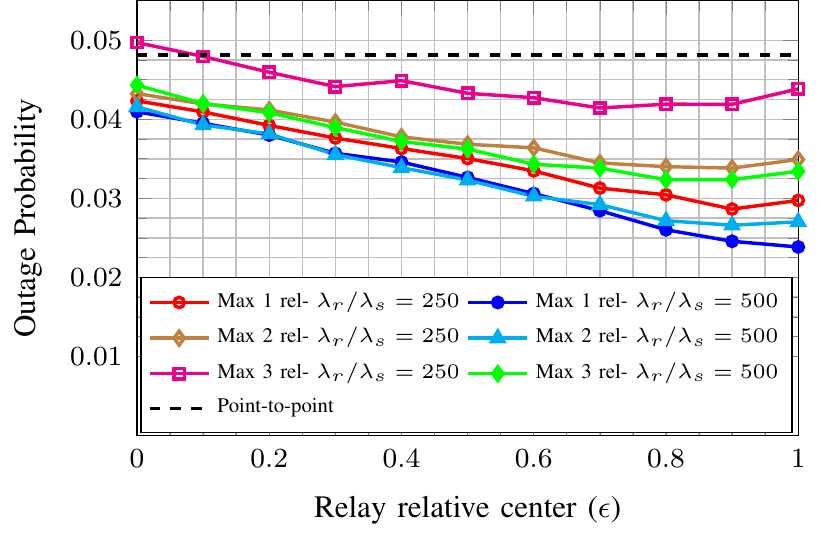}
	\caption{Outage probability for NNC as a function point around which the relays are chosen, one the line between source and destination, for different relative densities. $P_r/P_s=-10$dB, $\lambda_s =10^{-4}$, $R=1$b/use, $\alpha=4$. $d=10$. OP is optimized w.r.t. the noise compression variance.}
	\label{fig:NNC3}
\end{figure}

\begin{figure}
	\centering
\includegraphics[]{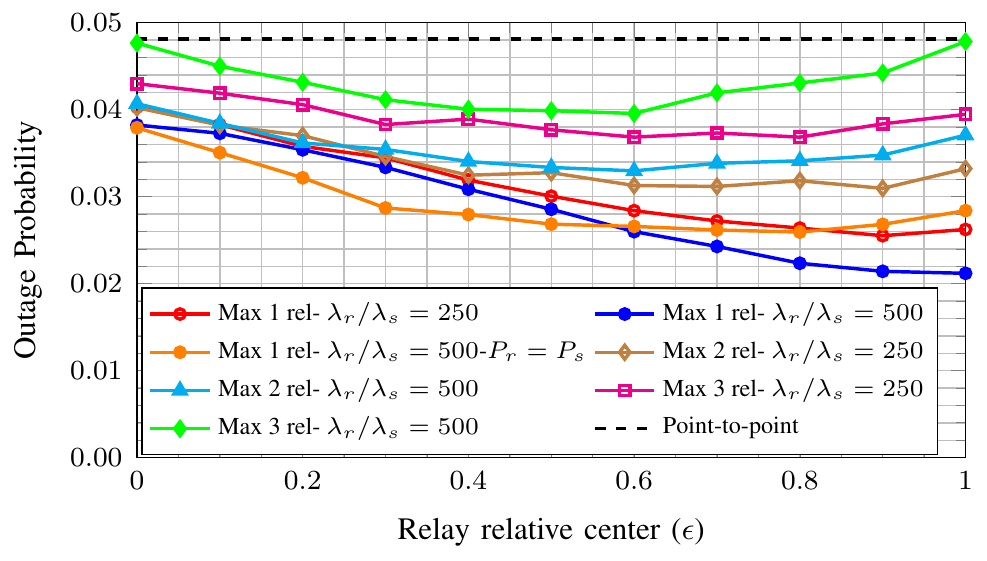}
	\caption{Outage probability for MNNC as a function point around which the relays are chosen, one the line between source and destination, for different relative densities. $P_r/P_s=-10$dB, $\lambda_s =10^{-4}$, $R=1$b/use, $\alpha=4$.  OP is optimized w.r.t. the noise compression variance.}
	\label{fig:MNNC3}
\end{figure} 





\subsection{Effect of Power Optimization and Interference-Aware Relays}


In this section we study the performance of the protocols and the optimal number of relays when the transmission power is optimized and when interference-aware relays are employed. When interference aware relays are used, we consider that only one of the thresholds (source-relay or source-destination) is used to activate the relays. In the previous section we saw that when fixed power is used, it is best to use a single relay. Since there are several parameters to consider for each protocol (number of relays, relay power, thresholds) and for space reasons, we focus first on MNNC which has the best performance. Afterwards we compare the performance of ODF and NNC to MNNC to see how close they are to MNNC when their respective parameters are optimized.
\begin{figure} [!t]
	\centering
\includegraphics[]{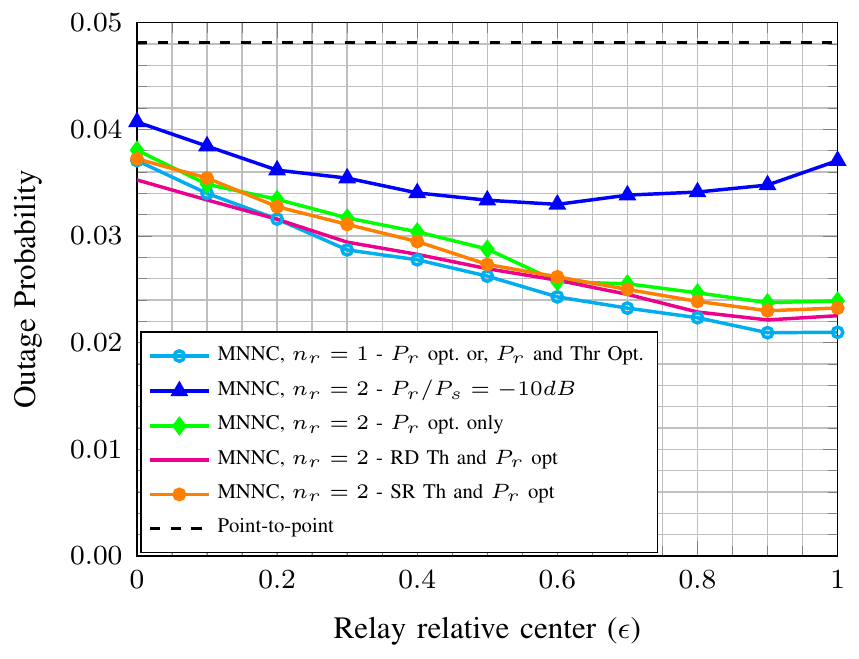}
	\caption{Outage probability for MNNC with optimized relay power and thresholds. $\lambda_r/\lambda_s=500$, $\lambda_s =10^{-4}$, $R=1$b/use, $\alpha=4$. }
	\label{fig:comp1}
\end{figure} 

In Fig. \ref{fig:comp1} we plot the OP of MNNC under different parameter optimizations when at most one or two relays are used. It can be seen that in the case of at most one relay ($n_{r}=1$), it is the same to optimize the relay power $P_r$ or to optimize at the same time the relay power and either of the thresholds, which implies that the thresholds are not necessary for this protocol in terms of OP, and that the relay could remain on all the time. In the case of at most two relays ($n_{r}=2$), optimizing only the relay power brings the OP very  close to that of using a single relay with optimized power. If, in addition, the thresholds are optimized as well, the performance becomes the same as using one relay. Nevertheless, the advantage of using two or more relays is that the transmission power of each relay can be reduced with respect to employing only one relay, as it is shown in Fig.~\ref{fig:optpowMNNC}.
\begin{figure} [!t]
	\centering
\includegraphics[]{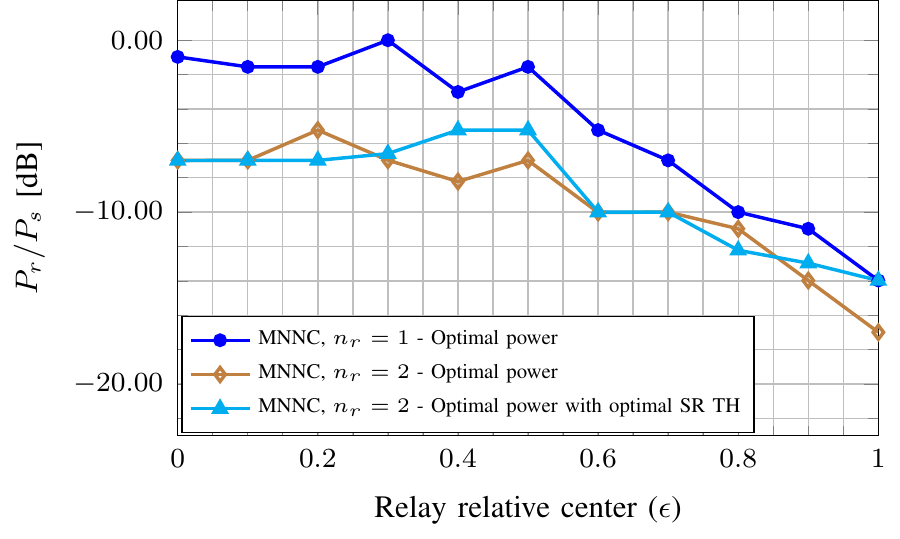}
	\caption{Optimal power for MNNC corresponding to the OP of Fig. \ref{fig:comp1}. Using more than one relay allows the reduction of the relay power. $\lambda_r/\lambda_s=500$, $\lambda_s =10^{-4}$, $R=1$b/use, $\alpha=4$.  }
	\label{fig:optpowMNNC}
\end{figure} 

Since the best performance with MNNC  is obtained by simply optimizing the relay power, in the following plots of ODF and NNC we use MNNC with $n_{r}=1$ and optimized $P_r$ as a benchmark comparison. In Fig. \ref{fig:comp2} we plot the OP of ODF for different number of relays when the power and the thresholds are optimized.
For the case of one relay, we see that by optimizing the transmission power only or both the transmission and the thresholds, a performance similar to that of MNNC can be achieved when the relay is chosen close to the source ($\epsilon < 0.4$), while the performance of MNNC is much better than that of ODF near the destination. On the other hand, and in contrast with MNNC, when at most two relays are used, the performance does not improve substantially by optimizing the relay transmission power or the thresholds. This is probably because the second relay is on average further and has a smaller probability of decoding the message of the source, and also because we assume that the relays outside the typical cluster are always on in the case of ODF even if they cannot decode.

Finally, in Fig. \ref{fig:comp3} we consider the optimization of the power and thresholds for NNC. We plot this together with the plots of MNNC and ODF for a single relay with optimized power. We optimized the power and the thresholds when using at most one or two relays, and found that the best option was to use an optimized RD threshold and power, though the gains were marginal compared to using a small fixed power. In the case of using up to two relays the performance improved notoriously compared to using a constant power, and the best option was to optimize both the transmission power and using an RD threshold. Similar to the case of MNNC, the performance obtained using two relays with optimized parameters was similar to those corresponding to a single relay, but each relay can use a smaller transmission power.

\begin{figure} [!t]
	\centering
\includegraphics[]{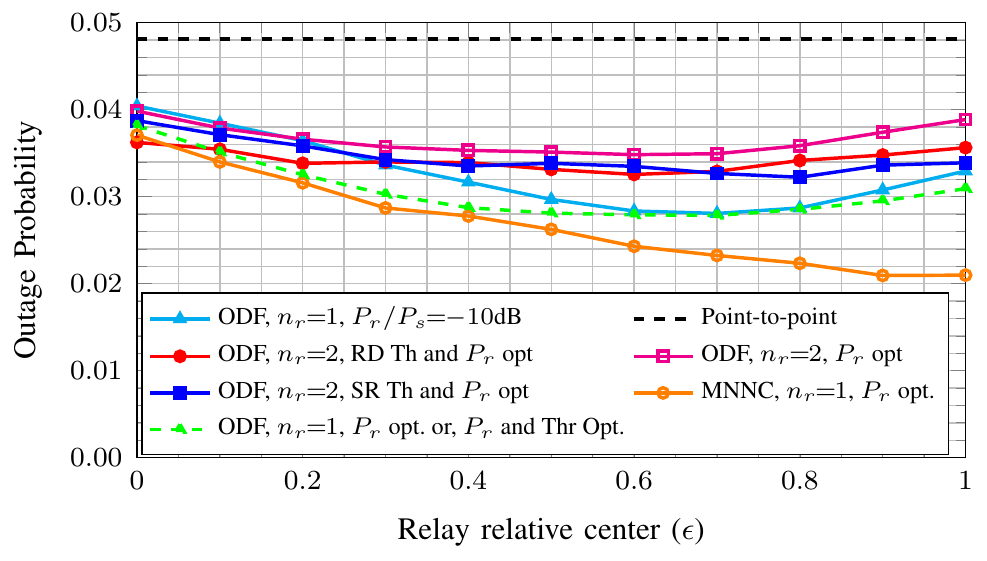}
	\caption{Outage probability for ODF with optimized relay power and thresholds, compared to the best performance obtained with MNNC. $\lambda_r/\lambda_s=500$, $\lambda_s =10^{-4}$, $R=1$b/use, $\alpha=4$.}
	\label{fig:comp2}
\end{figure} 

\begin{figure} [!t]
	\centering
\includegraphics[]{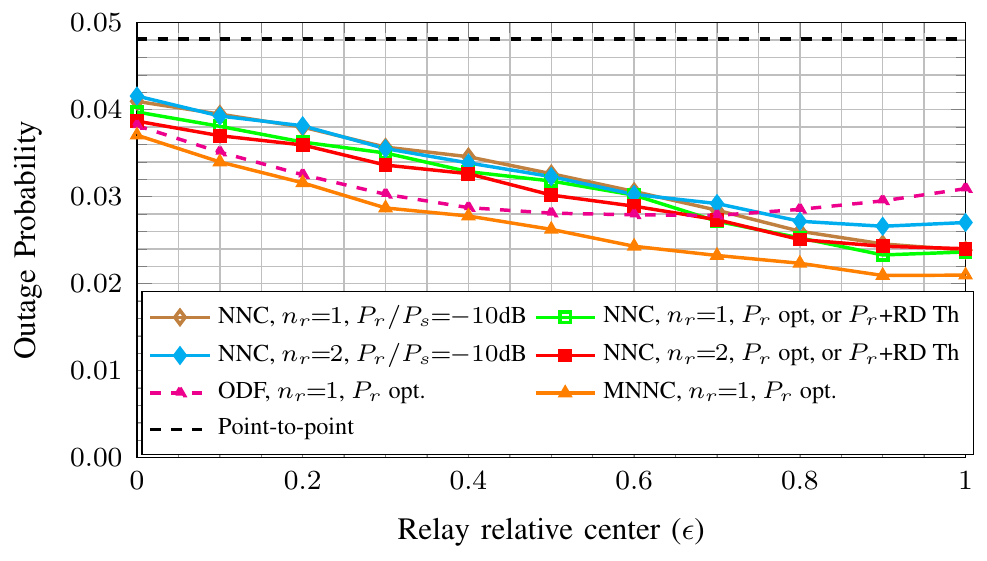}
	\caption{Outage probability for NNC with optimized relay power and thresholds, compared to the best performance obtained with MNNC and ODF. $\lambda_r/\lambda_s = 500$, $\lambda_s =10^{-4}$, $R=1$b/use, $\alpha=4$.  OP is optimized w.r.t. the noise compression variance.}
	\label{fig:comp3}
\end{figure} 
As a conclusion from this section we see that MNNC performs better than the other two protocols, which is reasonable since it can be interpreted as a combination of both, while ODF comes close when the relay is chosen near the source, and NNC comes closer provided that the relay is near the destination. In the three protocols we observed that using a single relay and optimizing the relay power is enough to attain the best performance in terms of the OP for each case. Furthermore, using more relays and optimizing the relay transmission power in the case of MNNC, or NNC, does not improve the OP compared to employing only one relay, but reduces the power consumption of each relay. In the case of ODF, using more relays does not reduce the OP even when the transmission power or thresholds are activated. This conclusion could be affected by the simplifying assumption that the relays in the other clusters are always active for ODF, even if they cannot decode.

%
%
%
%

\section{Summary and Concluding Remarks}  \label{sec:conclusions}

In this paper we studied the performance of some of the most advanced cooperative full-duplex relaying protocols, namely ODF, NNC and MNNC in the context of a large wireless network. The following observations can be made:
\begin{itemize}
\item In general, cooperation was most useful when the density of relays was much larger than the density of sources, i.e., $\lambda_r \gg \lambda_s$, and in the three cases a large reduction of the OP (around 50\%) was obtained by using a single relay with a fixed low transmission power (compared to the source).

\item MNNC was shown to outperform NNC and ODF in all cases. NNC without optimizations was shown to perform better than ODF without optimizations near the destination, while ODF was better near the midpoint between the source and destination.

\item The best performance of each protocol was obtained by using a single relay and optimizing its transmission power, and using a threshold based activation scheme was not necessary. In the case of MNNC and NNC, if two relays were used and the transmission power was optimized, the same performance was be obtained as with one relay but each relay could use a much smaller transmission power. In ODF using more than one relay did not improve the OP but this could be influenced by the simplifying assumption that the relays in other clusters are always active for ODF, even if they cannot decode, which results in a worst-case interference scenario.
\end{itemize}  
Although the derivation of closed form expressions for the outage probabilities corresponding to the investigated networks and protocols were discarded due to the involved mathematical complexity,  the approach taken in this paper aimed at drawing important conclusions with respect to the use of state of the art relaying protocoles in large wireless networks.

\section*{Acknowledgment}
The authors are grateful to Dr. B. Blaszczyszyn  for valuable and inspiring discussions at the early stage of this work. 

\bibliographystyle{IEEEtran}
\bibliography{IEEEabrv,relaygeo}
\end{document}